\documentclass[conference]{IEEEtran}

\IEEEoverridecommandlockouts

\usepackage{times}
\usepackage{epsfig}
\usepackage{graphicx}
\usepackage{amsmath}
\usepackage{amssymb}
\usepackage{amsthm}
\usepackage{enumerate}
\usepackage{indentfirst}
\usepackage{subfigure}
\usepackage{algorithmic}
\usepackage{algorithm}
\usepackage{url}
\usepackage{booktabs}
\usepackage{multirow}
\usepackage{longtable}
\usepackage{cite}

\theoremstyle{remark}

\newtheorem{lemma}{Lemma}
\newtheorem{corollary}{Corollary}

\long\def\symbolfootnote[#1]#2{\begingroup\def\thefootnote{\fnsymbol{footnote}}
\footnote[#1]{#2}\endgroup}

\ifCLASSINFOpdf
\else
\fi
\begin{document}
\title{Multicell Random Beamforming with CDF-based Scheduling: Exact Rate and Scaling Laws}

\author{\IEEEauthorblockN{Yichao Huang and Bhaskar D. Rao}
\IEEEauthorblockA{Department of Electrical and Computer Engineering\\
University of California San Diego, La Jolla, CA, 92093\\
Email: yih006@ucsd.edu; brao@ece.ucsd.edu}
\thanks{This research was supported by Ericsson endowed chair funds, the Center for Wireless Communications, UC Discovery grant com09R-156561 and NSF grant CCF-1115645.}}

\maketitle

\begin{abstract}
%\boldmath
In a multicell multiuser MIMO downlink employing random beamforming as the transmission scheme, the heterogeneous large scale channel effects of intercell and intracell interference complicate analysis of distributed scheduling based systems. In this paper, we extend the analysis in \cite{huang12j} and \cite{huang13j} to study the aforementioned challenging scenario. The cumulative distribution function (CDF)-based scheduling policy utilized in \cite{huang12j} and \cite{huang13j} is leveraged to maintain fairness among users and simultaneously obtain multiuser diversity gain. The closed form expression of the individual sum rate for each user is derived under the CDF-based scheduling policy. More importantly, with this distributed scheduling policy, we conduct asymptotic (in users) analysis to determine the limiting distribution of the signal-to-interference-plus-noise ratio, and establish the individual scaling laws for each user.
\end{abstract}

%\newpage
%\begin{IEEEkeywords}
%Rate scaling, CDF-based scheduling, distributed scheduling, multiuser diversity, extreme value theory
%\end{IEEEkeywords}

% \fi
\IEEEpeerreviewmaketitle

\section{Introduction}\label{introduction}
%\IEEEPARstart{T}{his} demo file is intended to serve as a ``starter file''
%for IEEE journal papers produced under \LaTeX\ using
%IEEEtran.cls version 1.7 and later.

%\IEEEpubidadjcol
With the emerging heterogeneous cellular structure \cite{damnjanovic11} and the ever shrinking cell size, achieving high capacity with low design complexity in a multicell multiuser MIMO downlink has drawn considerable interest in recent years, e.g, see \cite{gesbert10} and the references therein. Distributed scheduling policies are often favored due to operational scalability and affordable complexity incurred by the limited capacity of the backhaul. Under the employed distributed scheduling policy, analysis on the multicell network builds upon the extensive studies and insights drawn from the single cell network.

For the single cell network without intercell interference, capacity boosting scheme relies on the independent varying channels across users, i.e., the well known multiuser diversity gain \cite{knopp95}. To further harness this gain with multiple antennas, the notion of opportunistic beamforming is proposed \cite{viswanath02}, which is later extended to the notion of random beamforming \cite{sharif05} to have same sum capacity growth as nonlinear precoding schemes \cite{caire03, weingarten06} with reduced feedback requirement \cite{love08}. Multiuser diversity depends heavily on the scheduling policy, and it is important to guarantee scheduling fairness while achieving this gain in consideration of the heterogeneous large scale channel effects. This issue is tackled in \cite{huang13j} (the closed form sum rate in a homogeneous setup is derived in \cite{huang12l}), by leveraging the cumulative distribution function (CDF)-based scheduling policy \cite{park05} to satisfy the two desired features: multiuser diversity and scheduling fairness. According to the CDF-based scheduling policy \cite{park05}, each user can be equivalently viewed as competing with other users with the same CDF, and thus each user's rate is independent of the statistics of other users. This interesting property enables a ``micro" understanding of each user's rate performance compared to the conventional ``macro" understanding of the sum rate performance. Due to this property, the notion of individual sum rate and individual scaling laws are proposed in \cite{huang13j} to further understand both the exact and the asymptotic performance of random beamforming in a heterogeneous setup (and with different selective feedback schemes).

In a multicell network, the heterogeneous channel effects come naturally from the different experienced intercell interference across users, even for a SISO setup. This issue is investigated in \cite{huang12j} for a generic single antenna based heterogeneous multicell OFDMA network, with both exact rate expression and asymptotic rate approximation derived. In \cite{huang12c}, the rate of convergence and the individual scaling laws are established for the SISO multicell setup. Both \cite{huang12j} and \cite{huang12c} employ the CDF-based scheduling to maintain the two aforementioned scheduling features. In \cite{gesbert11}, the rate scaling for a power controlled network is examined with additional distance-based random
variables. Since the time variations for the large scale and small scale channel effects are vastly different \cite{choi08}, both \cite{huang12j} and \cite{huang12c} concentrate on the randomness of the small scale channel effects. In \cite{tajer11}, a normalized form of transformation is applied in a multicell network to achieve fairness. The main difference of using CDF-based scheduling \cite{park05} is the inherent nonlinear functional transformation to strictly guarantee user fairness.

In this work, we extend the analysis in \cite{huang12j} and \cite{huang13j} to a generic multicell multiuser MIMO setup. The random beamforming is utilized as the multi-antenna transmission scheme to reduce feedback need (for part of the literature survey regarding random beamforming, please refer to \cite{huang13j}). With spatial multiplexing in each cell, both intercell interference and intracell interference exist and users would experience heterogeneous interference. Under the CDF-based scheduling policy, we firstly derive the individual sum rate for each user from the exact analysis perspective. We further prove the type of convergence and the rate of convergence to the limiting distribution to establish the individual scaling laws.

\section{System Model}\label{system}
Consider the downlink of a generic multicell multiuser MIMO network. We assume a narrowband model and the established analysis in this work can be extended to the wideband model such as OFDMA using the techniques developed in \cite{huang13j}. Full spectrum reuse is assumed, and the process of cell association is assumed to be performed in advance. Without loss of generality, one base station $B_0$ equipped with $M$ antennas from the base station set $\mathcal{B}$ and its associated single antenna users $\mathcal{K}_0=\{1,\ldots,k,\ldots,K_0\}$ with $|\mathcal{K}_0|=K_0$ are considered. The random beamforming strategy at base station $B_0$ employs $M$ random orthonormal vectors $\boldsymbol{\phi}_m^{(0)}\in\mathbb{C}^{M\times 1}$ for $m=1,\ldots,M$, where the $\boldsymbol{\phi}_i^{(0)}$ are drawn from an isotropic distribution independently every $T$ (denoting the channel coherence interval for the block fading model) channel uses \cite{sharif05}. Denoting $s_m^{(0)}(t)$ as the $m$th transmission symbol at time $t$, the transmitted vector $\mathbf{s}^{(0)}(t)$ from base station $B_0$ at time $t$, is given as:
\begin{equation}\label{system_eq_1}
\mathbf{s}^{(0)}(t)=\sum_{m=1}^{M}\boldsymbol{\phi}_m^{(0)}(t)s_m^{(0)}(t), \quad t=1,\ldots,T.
\end{equation}

The received signal $y_{k}^{(0)}$ of user $k$ (the time variable $t$ is dropped for notational convenience) is represented by
\begin{align}\label{system_eq_2}
y_{k}^{(0)}&=\sum_{m=1}^{M}\sqrt{G_k^{(0)}}\mathbf{h}_{k}^{(0)}\boldsymbol{\phi}_m^{(0)}s_m^{(0)}+\sum_{b=1}^{J_k}\sum_{m=1}^{M}\sqrt{G_k^{(b)}}\mathbf{h}_{k}^{(b)}\boldsymbol{\phi}_m^{(b)}s_m^{(b)}\notag\\
&\quad\quad\quad +v_{k}^{(0)},\quad k\in \mathcal{K}_0,
\end{align}
where the superscript indicates the base station with $(0)$ being the cell of interest and $(b \neq 0)$ being the interfering cells. $J_k$
denotes the number of interfering cells for user $k$, and $v_{k}^{(0)}$ denotes the additive white noise distributed with $\mathcal{CN}(0,\sigma_k^2)$. $\mathbf{s}^{(0)}$ and $\mathbf{s}^{(b)}$ are the
transmitted symbols by the serving cell and the interfering cell $B_b$ with $\mathbb{E}\left[|\mathbf{s}^{(0)}|^2\right]=p_0$ and
$\mathbb{E}\left[|\mathbf{s}^{(b)}|^2\right]=p_b$. $\mathbf{h}_{k}^{(0)}\in\mathbb{C}^{1\times M}$ and $\mathbf{h}_{k}^{(b)}\in\mathbb{C}^{1\times M}$, which are assumed to be independent across users, denote the small scale channel gains between the serving cell and user $k$, and between the interfering cell $B_b$ and user $k$. $G_k^{(0)}$ and $G_k^{(b)}$ represent the large scale channel gains between the serving cell and user $k$, and between the interfering cell $B_b$ and user $k$ respectively. Based on the aforementioned assumption, denoting $Z_{k,m}^{(0)}$ as the $\mathsf{SINR}$ of user $k$ for beam $m$, it can be expressed as:
\small
\begin{align}
Z_{k,m}^{(0)}&=\frac{G_k^{(0)}p_0|\mathbf{h}_{k}^{(0)}\boldsymbol{\phi}_m^{(0)}|^2}{\mathop{\sum}\limits_{i\neq m}^{M}G_k^{(0)}p_0|\mathbf{h}_{k}^{(0)}\boldsymbol{\phi}_i^{(0)}|^2 + \mathop{\sum}\limits_{b=1}^{J_k}G_k^{(b)}p_b\mathop{\sum}\limits_{i=1}^{M}|\mathbf{h}_{k}^{(b)}\boldsymbol{\phi}_i^{(b)}|^2+M\sigma_k^2}\notag\\
\label{system_eq_3}&=\frac{\rho_k^{(0)}|\mathbf{h}_{k}^{(0)}\boldsymbol{\phi}_m^{(0)}|^2}{\mathop{\sum}\limits_{i\neq m}^{M}\rho_k^{(0)}|\mathbf{h}_{k}^{(0)}\boldsymbol{\phi}_i^{(0)}|^2+\mathop{\sum}\limits_{b=1}^{J_k}\rho_k^{(b)}\mathop{\sum}\limits_{i=1}^{M}|\mathbf{h}_{k}^{(b)}\boldsymbol{\phi}_i^{(b)}|^2+1},
\end{align}
\normalsize
where $\rho_k^{(0)}\triangleq \frac{G_k^{(0)}p_0}{M\sigma_k^2}$, $\rho_k^{(b)}\triangleq \frac{G_k^{(b)}p_b}{M\sigma_k^2}$.

Now we examine the statistics of $Z_{k,m}^{(0)}$. Note that the time variations for the large scale and small scale channel effects are vastly
different \cite{huang12j}. The variation of the small scale channel gain $\mathbf{h}$ occurs on the order of milliseconds; whereas the large scale channel gain $G$ which may consist of path loss, antenna gain, and shadowing, varies usually on the order of seconds. Therefore, $G$ is assumed to be known in advance by the system, through infrequent feedback; and the elements of $\mathbf{h}$ are modeled as complex Gaussian with zero mean and unit variance. Since for a given user $k$, the $Z_{k,m}^{(0)}$'s are identically distributed and correlated, the beam index $m$ can be dropped in the expression of the CDF, which is derived in the following lemma.

\begin{lemma} \label{lemma_1}
The CDF of $Z_k^{(0)}$ can be expressed as
\begin{equation}\label{system_eq_4}
F_{Z_k^{(0)}}(x)=1-\frac{e^{-\frac{x}{\rho_k^{(0)}}}\prod_{b=1}^{J_k}\left(\frac{\rho_k^{(0)}}{\rho_k^{(b)}}\right)^M}{(x+1)^{M-1}\prod_{b=1}^{J_k}\left(x+\frac{\rho_k^{(0)}}{\rho_k^{(b)}}\right)^M},\quad x\geq 0.
\end{equation}
\end{lemma}
\begin{proof}
The main technique relies on the use of the moment-generating function (MGF) \cite{proakis01}. Denote $\vartheta_k^{(0)}\triangleq \rho_k^{(0)}|\mathbf{h}_{k}^{(0)}\boldsymbol{\phi}_m^{(0)}|^2$, and $\zeta_k^{(0)}\triangleq\sum_{i\neq m}^{M}\rho_k^{(0)}|\mathbf{h}_{k}^{(0)}\boldsymbol{\phi}_i^{(0)}|^2+\sum_{b=1}^{J_k}\rho_k^{(b)}\sum_{i=1}^{M}|\mathbf{h}_{k}^{(b)}\boldsymbol{\phi}_i^{(b)}|^2$. Then the CDF of $Z_k^{(0)}$ can be derived using the following procedure:
\begin{align}
F_{Z_k^{(0)}}(x)&=\mathbb{P}\left(\frac{\vartheta_k^{(0)}}{\zeta_k^{(0)}+1}\leq x\right)\notag\\
&=\int_0^\infty\mathbb{P}\left(\vartheta_k^{(0)}\leq x(\zeta_k^{(0)}+1)\right)f_{\zeta_k^{(0)}}(y)dy\notag\\
\label{system_eq_5}&\mathop{=}\limits^{(a)}1-e^{-\frac{x}{\rho_k^{(0)}}}\int_0^\infty e^{-\frac{x\zeta_k^{(0)}}{\rho_k^{(0)}}}f_{\zeta_k^{(0)}}(y)dy,
\end{align}
where (a) follows from the fact that $\mathbb{P}\left(\vartheta_k^{(0)}\leq x(\zeta_k^{(0)}+1)\right)$ corresponds to the CDF of the exponential distribution at $x(\zeta_k^{(0)}+1)$. From (\ref{system_eq_5}), we note that the expression inside the integral corresponds to the MGF of $\zeta_k^{(0)}$, denoted as $\Psi_{\zeta_k^{(0)}}(\tau)$, at $-\frac{x}{\rho_k^{(0)}}$. Due to the additive effect reflected in $\zeta_k^{(0)}$, its MGF can be obtained below:
\begin{equation}\label{system_eq_6}
\Psi_{\zeta_k^{(0)}}(\tau)=\frac{1}{(1-\rho_k^{(0)}\tau)^{M-1}}\prod_{b=1}^{J_k}\frac{1}{(1-\rho_k^{(b)}\tau)^M}.
\end{equation}
Combing (\ref{system_eq_5}) and (\ref{system_eq_6}) yields $F_{Z_k^{(0)}}$ expressed in (\ref{system_eq_4}).
\end{proof}
The $\mathsf{SINR}$ will be fed back\footnote[1]{Full feedback wherein each user feeds back the $\mathsf{SINR}$ for $M$ beams is assumed. The established results in this paper can be extended to different forms of selective feedback using the techniques developed in \cite{huang13j}.} and used for scheduling, which is pursued next.

\begin{figure*}[htb]
\begin{equation*}\label{scheduling_eq_8}
\psi_{k,j}^{(0)}=\frac{1}{((M-1)(\ell+1)-j))!}\frac{d^{(M-1)(\ell+1)-j}}{dx^{(M-1)(\ell+1)-j}}\left[\frac{1}{\prod_{b=1}^{J_k}(x+\frac{\rho_k^{(0)}}{\rho_k^{(b)}})^{M(\ell+1)}}\right]\Bigg|_{x=-1}
\end{equation*}
\begin{equation*}\label{scheduling_eq_9}
\psi_{k,j}^{(b)}=\frac{1}{(M(\ell+1)-j))!}\frac{d^{M(\ell+1)-j}}{dx^{M(\ell+1)-j}}\left[\frac{1}{(x+1)^{(M-1)(\ell+1)}\prod_{q\neq b}^{J_k}(x+\frac{\rho_k^{(0)}}{\rho_k^{(q)}})^{M(\ell+1)}}\right]\Bigg|_{x=-\frac{\rho_k^{(0)}}{\rho_k^{(q)}}}
\end{equation*}
\rule{\linewidth}{0.5pt}
\end{figure*}

\section{CDF-based Scheduling Policy \\and Individual Sum Rate}\label{scheduling}
After receiving the $\mathsf{SINR}_{k,m}^{(0)}$ from user $k$ for beam $m$, the scheduler is ready to perform scheduling.  In order to guarantee scheduling fairness and obtain multiuser diversity gain, we employ the CDF-based scheduling policy \cite{park05} as the
opportunistic scheduling policy\footnote[2]{For detailed motivation of the CDF-based scheduling as well as the rationale behind the notion of individual sum rate and individual scaling laws, please refer to \cite{huang12j} and \cite{huang13j}.}. According to this policy, the scheduler will utilize $F_{Z_k^{(0)}}$, and performs the following functional transformation:
\begin{equation}\label{scheduling_eq_1}
\tilde{Z}_{k,m}^{(0)}=F_{Z_k^{(0)}}\left(Z_{k,m}^{(0)}\right).
\end{equation}

The transformed random variable $\tilde{Z}_{k,m}^{(0)}$ is uniformly distributed ranging from $0$ to $1$, and can be regarded as the virtual
received $\mathsf{SINR}$ of user $k$ for beam $m$. The transformed random variables $\tilde{Z}_{k,m}^{(0)}$'s are i.i.d. across users for a given beam, which enables the maximization at the scheduler side to perform scheduling in a fair manner. Denoting $k_m^*$ as the random variable representing the selected user for beam $m$, then:
\begin{equation} \label{scheduling_eq_2}
k_m^*=\arg\max_{k\in\mathcal{K}_0}\;\tilde{Z}_{k,m}^{(0)}.
\end{equation}
After user $k_m^*$ is selected per (\ref{scheduling_eq_2}), the scheduler utilizes the corresponding $Z_{k_m^*,m}^{(0)}$ for rate matching of the selected user. Let $X_m^{(0)}$ be the $\mathsf{SINR}$ of the selected user for beam $m$, and now consider the sum rate for a given base station $B_0$ defined as follows:
\begin{equation} \label{scheduling_eq_3}
R^{(0)}=\mathbb{E}\left[\sum_{m=1}^M\log\left(1+X_m^{(0)}\right)\right].
\end{equation}
From the aforementioned analysis, the sum rate can be formulated as:
\begin{align}
R^{(0)}&\mathop{\simeq}\limits^{(a)}M\mathbb{E}_{k_m^*}\left[\int_0^1\log\left(1+F_{Z_{k_m^*,m}^{(0)}}^{-1}(x)\right)
dx^{K_0}\right]\notag\\
\label{scheduling_eq_4}&\mathop{=}\limits^{(b)}\frac{M}{K_0}\sum_{k=1}^{K_0}\int_0^{\infty}\log(1+t) d(F_{Z_{k}^{(0)}}(t))^{K_0},
\end{align}
where (a) follows from the sufficient small probability that multiple beams are assigned to the same user;(b) follows from the change of variable $x=F_{Z_{k_m^*,m}^{(0)}}(t)$ and the fairness property of the CDF-based scheduling policy.

The CDF-based scheduling enables a ``micro" level understanding of each user's rate performance, from both exact and asymptotic perspective. The individual sum rate for user $k$ is defined in \cite{huang13j} as the the individual user rate multiplied by the number of users in cell $B_0$, namely:
\begin{equation} \label{scheduling_eq_5}
\hat{R}_k^{(0)}\triangleq K_0R_k^{(0)}=M\int_0^{\infty}\log(1+x) d(F_{Z_{k}^{(0)}}(x))^{K_0}.
\end{equation}

Now employing the steps described below, the closed form expression for the individual sum rate $\hat{R}_k^{(0)}$ is derived.

\textit{Step 1 (PDF Decomposition):} This step, utilized similarly in \cite{huang12j} and \cite{huang13j}, is the essential step to decompose $d(F_{Z_{k}^{(0)}}(x))^{K_0}$ into the following amenable input for Step 2.
\begin{align}\label{scheduling_eq_6}
&d(F_{Z_{k}^{(0)}}(x))^{K_0}=K_0\sum_{\ell=0}^{K_0-1}{K_0-1\choose \ell}\frac{(-1)^{\ell}}{\ell+1}\notag\\
&\quad\times d\left(1-\left(\frac{e^{-\frac{x}{\rho_k^{(0)}}}(\frac{\rho_k^{(0)}}{\rho_k^{(b)}})^M}{(x+1)^{M-1}\prod_{b=1}^{J_k}(x+\frac{\rho_k^{(0)}}{\rho_k^{(b)}})^M}\right)^{\ell+1}\right).
\end{align}

\textit{Step 2 (Partial Fraction Expansion \cite{gradshteyn07}):} This step manipulates part of (\ref{scheduling_eq_6}) for further integration.
\begin{align}\label{scheduling_eq_7}
&\frac{1}{(x+1)^{(M-1)(\ell+1)}\prod_{b=1}^{J_k}\left(x+\frac{\rho_k^{(0)}}{\rho_k^{(b)}}\right)^{M(\ell+1)}}\notag\\
&\quad=\sum_{j=1}^{(M-1)(\ell+1)}\frac{\psi_{k,j}^{(0)}}{(x+1)^j} +\sum_{b=1}^{J_k}\sum_{j=1}^{M(\ell+1)}\frac{\psi_{k,j}^{(b)}}{\left(x+\frac{\rho_k^{(0)}}{\rho_k^{(b)}}\right)^j},
\end{align}
where the expressions for $\psi_{k,j}^{(0)}$ and $\psi_{k,j}^{(b)}$ are given on top of this page.

Combining the outcomes of the two aforementioned steps, we can derive the closed form expression for $\hat{R}_k^{(0)}$ in the following procedure:
\small
\begin{align}\label{scheduling_eq_8}
&\hat{R}_k^{(0)}=\frac{MK_0}{\ln 2}\sum_{\ell=0}^{K_0-1}{K_0-1\choose \ell}\frac{(-1)^{\ell}}{\ell+1}\notag\\
&\; \times \int_0^\infty \ln(1+x)d\left(1-\left(\frac{e^{-\frac{x}{\rho_k^{(0)}}}(\frac{\rho_k^{(0)}}{\rho_k^{(b)}})^M}{(x+1)^{M-1}\prod_{b=1}^{J_k}(x+\frac{\rho_k^{(0)}}{\rho_k^{(b)}})^M}\right)^{\ell+1}\right)\notag\\
&\;=\frac{MK_0}{\ln 2}\sum_{\ell=0}^{K_0-1}{K_0-1\choose \ell}\frac{(-1)^{\ell}}{\ell+1}\prod_{b=1}^{J_k}\left(\frac{\rho_k^{(0)}}{\rho_k^{(b)}}\right)^M\notag\\
&\quad\times \Bigg[\sum_{b=1}^{J_k}\sum_{j=1}^{M(\ell+1)}\psi_{k,j}^{(b)}\mathcal{I}_1\left(\frac{\ell+1}{\rho_k^{(0)}},\frac{\rho_k^{(0)}}{\rho_k^{(b)}},j\right)\notag\\
&\quad\;+\sum_{j=1}^{(M-1)(\ell+1)}\psi_{k,j}^{(0)}\mathcal{I}_2\left(\frac{\ell+1}{\rho_k^{(0)}},1,j+1\right)\Bigg],
\end{align}
\normalsize
where $\mathcal{I}_1(\alpha,\beta,\gamma)\triangleq \int_0^{\infty}\frac{e^{-\alpha x}}{(1+x)(\beta+x)^{\gamma}} dx$, and $\mathcal{I}_2(\alpha,\beta,\gamma)\triangleq
\int_0^{\infty}\frac{e^{-\alpha x}}{(\beta+x)^{\gamma}} dx$. The calculation for $\mathcal{I}_1(\alpha,\beta,\gamma)$ and $\mathcal{I}_2(\alpha,\beta,\gamma)$ has been discussed in \cite{huang12j}, and their closed form expressions can be found in \cite[(42)]{huang12j} and \cite[(43)]{huang12j}.

Up to now, we have performed exact analysis and derived the closed form results for the individual sum rate for an arbitrary selected user in a given base station. The derived results extend the exact analysis in \cite{huang12j} (multicell SISO setup) and \cite{huang13j} (single cell random beamforming), and can serve as a theoretical reference for evaluating the system performance under the CDF-based scheduling policy. In the next section, we will perform asymptotic analysis to evaluate the rate scaling laws for $\hat{R}_k^{(0)}$, which helps in understanding
the asymptotic behavior of an individual user.

\section{Individual Scaling Laws}\label{scaling}
This section is devoted to the asymptotic analysis. Section \ref{type} shows the type of convergence of $Z_k^{(0)}$. In Section \ref{rate}, the
convergence rate to the limiting distribution is studied to establish the individual rate scaling laws.

\subsection{Type of Convergence}\label{type}
Firstly, we need to examine the  tail behavior of the statistics of $Z_k^{(0)}$, which has the form presented in
(\ref{system_eq_4}). Tools from extreme value theory \cite{galambos78} are to be utilized. The following lemma describes the tail behavior
of $F_{Z_k^{(0)}}$.

\begin{lemma} \label{lemma_2}
$F_{Z_k^{(0)}}$ belongs to the domain of attraction of the Gumbel distribution, i.e., $F_{Z_k^{(0)}}\in\mathcal{D}(G_3)$.
\end{lemma}
\begin{proof}
(Sketch) In order to prove that $F_{Z_k^{(0)}}\in\mathcal{D}(G_3)$, it must be shown that
$\mathop{\lim}\limits_{x\rightarrow\infty}\frac{d}{dx}\left[\frac{1-F_{Z_k^{(0)}}(x)}{f_{Z_k^{(0)}}(x)}\right]=0$ \cite{galambos78}. The equivalent condition is: $\mathop{\lim}\limits_{x\rightarrow\infty}\frac{\left(F_{Z_k^{(0)}}(x)-1\right)f_{Z_k^{(0)}}'(x)}{(f_{Z_k^{(0)}}(x))^2}=1$. Since similar methodologies in proving \cite[Corollary 1]{huang12j} can be employed, we omit the detailed proof.
\end{proof}

\subsection{Rate of Convergence}\label{rate}
Knowing the type of convergence can lead to asymptotic approximation for the individual sum rate, which is investigated in \cite{huang12j}. Herein, dealing with higher order moments of the extreme order statistics and the rate of convergence is of interest. To establish the convergence rate to the limiting distribution for an individual user, the following definition of the so called growth function \cite{galambos78} is needed:
\begin{equation}\label{scaling_eq_5}
g_{Z_k^{(0)}}(x)\triangleq\frac{1-F_{Z_k^{(0)}}(x)}{f_{Z_k^{(0)}}(x)}.
\end{equation}
One important step in proving the
rate of convergence is to solve for a suitable coefficient sequence $w_{k:K_0}$ by solving the following equation:
\begin{equation}\label{scaling_eq_6}
1-F_{Z_k^{(0)}}(w_{k:K_0})=\frac{1}{K_0}.
\end{equation}

Due to the complicated form of $F_{Z_k^{(0)}}$ in (\ref{system_eq_4}), we need to find upper and lower bound for $w_{k:K_0}$ by constructive methods, namely, $w_{k:K_0}^{\mathsf{LB}}\leq w_{k:K_0}\leq w_{k:K_0}^{\mathsf{UB}}$, by deriving upper and lower bound for $F_{Z_k^{(0)}}$. Define $b_k^{\mathsf{min}}=\arg\mathop{\min}\limits_{0\leq b \leq J_k} \rho_k^{(b)}$, and $b_k^{\mathsf{max}}=\arg\mathop{\max}\limits_{0\leq b \leq
J_k} \rho_k^{(b)}$. Then, one suitable upper bound and lower bound for $F_{Z_k^{(0)}}$, denoted by $F_{Z_k^{(0)}}^{\mathsf{UB}}$ and $F_{Z_k^{(0)}}^{\mathsf{LB}}$, are presented in the following lemma.

\begin{lemma} \label{lemma_3}
\begin{equation}\label{scaling_eq_7}
F_{Z_k^{(0)}}^{\mathsf{UB}}(x)
=1-\frac{e^{-\frac{x}{\rho_k^{(0)}}}}{\left(1+\frac{\rho_k^{(b_k^{\mathsf{min}})}}{\rho_k^{(0)}}x\right)^{(J_k+1)M-1}},\quad x\geq 0,
\end{equation}
\begin{equation}\label{scaling_eq_8}
F_{Z_k^{(0)}}^{\mathsf{LB}}(x)
=1-\frac{e^{-\frac{x}{\rho_k^{(0)}}}}{\left(1+\frac{\rho_k^{(b_k^{\mathsf{max}})}}{\rho_k^{(0)}}x\right)^{(J_k+1)M-1}},\quad x\geq 0.
\end{equation}
\end{lemma}
\begin{proof}
(Sketch) The upper and lower bound can be constructed by examining the large scale channel effects of intercell interference for user $k$. Herein, we only provide an intuitive explanation. $F_{Z_k^{(0)}}^{\mathsf{UB}}$ can be obtained by assuming that the intercell and intracell interference has the same large scale channel effects $\rho_k^{(b_k^{\mathsf{min}})}$. The
$F_{Z_k^{(0)}}^{\mathsf{LB}}$ can be found by using a similar line of argument.
\end{proof}

Employing Lemma \ref{lemma_3}, the upper and lower bound for $w_{k:K_0}$ can be derived and are provided in the following corollary.

\begin{corollary} \label{corollary_1}
\small
\begin{align}
w_{k:K_0}^{\mathsf{UB}}&=\rho_k^{(0)}\log K_0-\rho_k^{(0)}((J_k+1)M-1)\log\left(\rho_k^{(b_k^{\mathsf{max}})}\log K_0\right)\notag\\
\label{scaling_eq_9}&\quad\quad+O(\log\log\log K_0).
\end{align}
\begin{align}
w_{k:K_0}^{\mathsf{LB}}&=\rho_k^{(0)}\log K_0-\rho_k^{(0)}((J_k+1)M-1)\log\left(\rho_k^{(b_k^{\mathsf{min}})}\log K_0\right)\notag\\
\label{scaling_eq_10}&\quad\quad+O(\log\log\log K_0).
\end{align}
\normalsize
\end{corollary}
\begin{proof}
The $w_{k:K_0}^{\mathsf{UB}}$ can be obtained via solving $1-F_{Z_k^{(0)}}^{\mathsf{LB}}(w_{k:K_0}^{\mathsf{UB}})=\frac{1}{K_0}$. Substituting
the expression of $F_{Z_k^{(0)}}^{\mathsf{LB}}$ and taking the $\log$ operator of both sides yields:
\begin{equation}\label{scaling_eq_11}
\frac{w_{k:K_0}^{\mathsf{UB}}}{\rho_k^{(0)}}+((J_k+1)M-1)\log\left(1+\frac{\rho_k^{(b_k^{\mathsf{max}})}}{\rho_k^{(0)}}w_{k:K_0}^{\mathsf{UB}}\right)=\log
K_0.
\end{equation}
Since $w_{k:K_0}^{\mathsf{UB}}\rightarrow\infty$ as $K_0\rightarrow\infty$, $\frac{w_{k:K_0}^{\mathsf{UB}}}{\rho_k^{(0)}}$ dominates
(\ref{scaling_eq_11}) and so we have $w_{k:K_0}^{\mathsf{UB}}\sim\rho_k^{(0)}\log K_0$. Then the sequence of $w_{k:K_0}^{\mathsf{UB}}$ can be
further written as $w_{k:K_0}^{\mathsf{UB}}=\rho_k^{(0)}\log K_0+d_{k:K_0}^{\mathsf{UB}}$, where $d_{k:K_0}^{\mathsf{UB}}\in o(\log K_0)$.

The expression for $d_{k:K_0}^{\mathsf{UB}}$ can be derived by substituting the formulation of $w_{k:K_0}^{\mathsf{UB}}$ into (\ref{scaling_eq_11}), which solves $d_{k:K_0}^{\mathsf{UB}}$ as
\small
\begin{align}
&d_{k:K_0}^{\mathsf{UB}}=-\rho_k^{(0)}((J_k+1)M-1)\log\left(\rho_k^{(b_k^{\mathsf{max}})}\log K_0\right)\notag\\
\label{scaling_eq_13}&\quad
-\rho_k^{(0)}((J_k+1)M-1)\log\left(1+\frac{\rho_k^{(0)}+\rho_k^{(b_k^{\mathsf{max}})}d_{k:K_0}^{\mathsf{UB}}}{\rho_k^{(0)}\rho_k^{(b_k^{\mathsf{max}})}\log
K_0}\right).
\end{align}
\normalsize

Therefore, $w_{k:K_0}^{\mathsf{UB}}$ exhibits a scaling performance as expressed in (\ref{scaling_eq_9}). The corresponding analysis for
$w_{k:K_0}^{\mathsf{LB}}$ can be conducted following the same line of arguments.
\end{proof}

Using the results from Corollary \ref{corollary_1} and observing the fact that $\log \rho_k^{(b_k^{\mathsf{max}})}$ or $\log
\rho_k^{(b_k^{\mathsf{min}})}$ are inconsequential when $K_0$ goes large, we have the following expression for $w_{k:K_0}$:
\begin{align}\label{scaling_eq_14}
w_{k:K_0}&=\rho_k^{(0)}\log K_0-\rho_k^{(0)}((J_k+1)M-1)\log\log K_0\notag\\
&\quad +O(\log\log\log K_0).
\end{align}

Once the expression of $w_{k:K_0}$ is obtained, we can have the following inequality for the selected user's $\mathsf{SINR}$ for beam $m$ (denoted by
$X_m^{(0)}$ in Section \ref{scheduling}) by employing \cite[Corollary A.1]{sharif05}, as follows:
%\small
\begin{align}
&\mathbb{P}\big\{\rho_k^{(0)}\log K_0-\rho_k^{(0)}(J_k+1)M\log\log K_0+O(\log\log\log K_0)\notag\\
&\quad\leq X_m^{(0)} \leq \rho_k^{(0)}\log K_0-\rho_k^{(0)}((J_k+1)M-2)\log\log K_0\notag\\
\label{scaling_eq_15} &\quad +O(\log\log\log K_0)\big\}\geq 1-O\left(\frac{1}{\log K_0}\right).
\end{align}
%\normalsize

Now we can state the following theorem for the rate scaling performance of $\hat{R}_k^{(0)}$.

\begin{corollary} \label{corollary_2}
\begin{equation}\label{scaling_eq_16}
\lim_{K_0\rightarrow\infty}\frac{\hat{R}_k^{(0)}}{M\log\log K_0}=1.
\end{equation}
\end{corollary}
\begin{proof}
(Sketch)
\begin{align}
&\hat{R}_k^{(0)}\leq M\mathbb{P}\left\{X_m^{(0)}\leq w_{k:K_0}+\rho_k^{(0)}\log\log K_0\right\}\notag\\
&\quad\quad\times\log\left(1+w_{k:K_0}+\rho_k^{(0)}\log\log K_0\right)\notag\\
&\; + M\mathbb{P}\left\{X_m^{(0)}\geq w_{k:K_0}+\rho_k^{(0)}\log\log K_0\right\}\log(1+\rho_k^{(0)}K_0)\notag\\
\label{scaling_eq_17}& \quad\quad \leq M\log\left(1+w_{k:K_0}+\rho_k^{(0)}\log\log K_0\right)+O(1).
\end{align}
A lower bound for $\hat{R}_k^{(0)}$ can be obtained similarly. Thus the individual scaling law for $\hat{R}_k^{(0)}$ exhibits a $M\log\log K_0$
growth in the large user regime.
\end{proof}

\textit{Remark:} Corollary \ref{corollary_2} informs us that after we bring in the notion of individual scaling law corresponding to the individual sum rate for an arbitrary selected user, it exhibits a $M\log\log K_0$ growth. This property is desirable from the perspective of opportunistic scheduling (e.g., greedy scheduling). In addition to this property, with the nonlinear functional transformation, the CDF-based scheduling policy can guarantee scheduling fairness in terms of long term user fairness (i.e., each user is equiprobable to be scheduled irrespective of their intercell and intracell interference).

\section{Conclusion}\label{conclusion}
The analytical impact of CDF-based scheduling policy in two special scenarios (multicell multiuser SISO and single cell multiuser MIMO) has been investigated in \cite{huang12j} and \cite{huang13j}. This work extends and generalizes our previous works by addressing the rate performance in a generic multicell multiuser MIMO downlink, with random beamforming as the signal transmission scheme. The most challenging part of this generalization lies in the existence of both intercell and intracell interference, as well as their accompanying heterogeneous large scale channel effects. The CDF-based scheduling helps us to deal with this challenging scenario, and enables a ``micro" understanding of the rate performance for any selected user. The closed form individual sum rate is derived employing the MGF and the PDF decomposition. With the constructed bounding technique, we also establish the individual scaling laws to show that CDF-based scheduling exhibits the same scaling performance as opportunistic scheduling (but achieves scheduling fairness additionally).

% or


\appendices


%\section{}\label{appenB}

%\section*{Acknowledgment}


\ifCLASSOPTIONcaptionsoff
  \newpage
\fi

\end{document}